\documentclass{eptcs}

\usepackage{amsmath}
\usepackage{amsthm, amssymb}
\usepackage{color}
\usepackage{breakurl}

\newcommand{\cf}{cf.}

\newcommand{\wrt}{w.r.t.}

\begin{document}

\newtheorem{example}{Example}
\newtheorem{definition}{Definition}
\newtheorem{lemma}{Lemma}
\newtheorem{theorem}{Theorem}

\title{Termination of Rewriting with and Automated Synthesis of Forbidden Patterns}

\author{Bernhard Gramlich\\
TU Vienna,
\url{gramlich@logic.at}\\
\and
Felix Schernhammer\thanks{The author has been supported by the Austrian Academy of Sciences under grant 22.361.}\\
TU Vienna, 
\url{felixs@logic.at}\\
}
\def\titlerunning{Forbidden Patterns: Termination and Automated Synthesis}
\def\authorrunning{B.\ Gramlich \& F.\ Schernhammer}

\maketitle

\begin{abstract}
We introduce a modified version of the well-known dependency pair
framework that is suitable for the termination analysis of rewriting
under forbidden pattern restrictions. By attaching contexts to dependency
pairs that represent the calling contexts of the corresponding
recursive function calls, it is possible to incorporate the
forbidden pattern restrictions in the (adapted) notion of dependency
pair chains, thus yielding a sound and complete approach to termination
analysis.
Building upon this contextual dependency pair framework we introduce
a dependency pair processor that simplifies problems by analyzing the contextual
information of the dependency pairs. Moreover, we show how this processor
can be used to synthesize forbidden patterns suitable for a given 
term rewriting system
on-the-fly during the termination analysis.
\end{abstract}

\section{Introduction and Overview}

Rewriting with \emph{forbidden patterns} \cite{wrs09} is a proper
restriction of term rewriting where subterms of terms may be forbidden
for reduction 
whenever they appear in a certain context and have a certain shape.
The main goal of rewriting with restrictions is to allow reductions
that are essential for computing results (i.e., normal forms) and to
disallow reductions that are not needed and may lead to 
infinite computations.

In \cite{wrs09} first criteria for completeness and termination of
rewriting 
with forbidden patterns were introduced. Here, by \emph{completeness}
we 
mean the power of restricted rewriting to compute useful results, which
in \cite{wrs09} were head-normal forms. The termination criterion
of \cite{wrs09} is based on a transformation from rewrite systems
with forbidden patterns into ordinary TRSs such that termination of
both coincides. 

In this work we provide another direct termination proof approach
based on a \emph{contextual} extension of the well-known dependency
pairs (DP) approach of \cite{ag00}, \cf\ also 
\cite{fsttcs06,lpar08,ic10-alarcon-gutierrez-lucas}. The basic idea is to enrich dependency pairs
by an additional component. This component is the calling context corresponding
to the recursive function call the dependency pair originated from.
Hence, the full contextual information is incorporated into
the dependency pairs and can be used to define an adequate notion
of chain respecting the restrictions imposed by forbidden patterns.

Building upon this generalized notion of dependency pair chains we
propose a DP processor that may simplify DP problems by analyzing
the contexts attached to dependency pairs. The processor
analyzes sequences of dependency pairs for being (potential)
DP chains, by checking whether the necessary reduction steps
are allowed in the respective contexts according to the forbidden
pattern restrictions. If it finds that a certain DP cannot occur
in any DP chain then this pair can safely be deleted from the DP
problem in question.

Our new approach is applicable to 
a wider class of rewrite systems with forbidden patterns than the
transformational approach of \cite{wrs09}. In addition it turns
out that as a byproduct of termination analysis in our framework
we get a method of synthesizing forbidden patterns suitable for
a given rewrite system on-the-fly and fully automatically. 

In order to evaluate our approach we used our method to analyze
the TRSs in the outermost category of the TPDB for outermost
termination. This makes sense, as outermost rewriting is a special
case of rewriting with forbidden patterns. Hence, our methods are
applicable. With the methods described in this work, the results
are already promising and better than some transformational approaches.
However, the potential of the contextual dependency pair approach
seems even larger given the results of the experiments of \cite{wst10}
where a more sophisticated DP processor was used for analyzing contexts
of dependency pairs. There the number of positive termination proofs
could almost be doubled. 

\section{Preliminaries} 
\label{sec_termination}

We assume a basic knowledge of and familiarity with the notions and
notations of rewriting as can be found, for instance, in \cite{bn98}.

The set of positions of a term $t$ is denoted by $Pos(t)$. For a
signature $\mathcal{F}$ the set of function symbol positions of $t$ is
denoted by $Pos_\mathcal{F}(t)$ and for a subsignature $\mathcal{F}'$ of
$\mathcal{F}$ by $Pos_{\mathcal{F}'}(t)$ we denote those positions $q$ of $Pos_{\mathcal{F}}(t)$
where $root(t|_q) \in \mathcal{F}'$. In reduction steps we sometimes
specify information about where the step takes place, as e.g.\ in 
$s \overset{p}{\rightarrow} t$, 
$s \overset{\leq p}{\rightarrow} t$, or
$s \overset{\not\leq p}{\rightarrow} t$.
We say that a rewrite rule
$l \rightarrow r$ overlaps a (variable-disjoint) term $t$ at
non-variable position $p\in Pos_\mathcal{F}(t)$ if $l$ and 
$t|_p$ unify.

A \emph{forbidden pattern} is a triple $\langle t, p, \lambda
\rangle$, consisting of a term $t$, a position $p \in Pos(t)$ and a
flag $\lambda \in \{h, b, a\}$. 
Given a term $s$ and a forbidden pattern $\pi = \langle t, p, \lambda \rangle$, 
$t$ and $p$ determine a set of positions $P_{t, p}(s) \subseteq Pos(s)$ by
$q \in P_{t, p}(s) \Leftrightarrow s|_o = t \sigma \wedge q =
o.p$ for some substitution $\sigma$
and some position $o$.
Moreover, for $\pi = \langle t, p, \lambda \rangle$, 
$P_{\pi}(s) = \{o \in Pos(s) \mid \exists q \in P_{t, p}(s) \colon o < q\}$ if $\lambda = a$, 
$P_{\pi}(s) = \{o \in Pos(s) \mid \exists q \in P_{t, p}(s) \colon o > q\}$ if $\lambda = b$ and
$P_{\pi}(s) = P_{t, p}(s)$ if $\lambda = h$. Given a set of forbidden patterns $\Pi$,
the set of \emph{forbidden} positions $\overline{Pos}^{\Pi}(s)$ w.r.t.\ $\Pi$ of a term $s$ is 
$\bigcup_{\pi \in \Pi}P_{\pi}(s)$. 
The \emph{allowed} positions $Pos^{\Pi}$ of $s$ (\wrt\ $\Pi$) are $Pos(s)
  \setminus \overline{Pos}^{\Pi}(s)$.
Rewriting with forbidden patterns (we write $\rightarrow_{\mathcal{R}, \Pi}$, 
or just $\rightarrow_{\Pi}$ 
-- or even only $\Pi$ as in
$\Pi$-termination --
if $\mathcal{R}$ is clear from the context)
is rewriting at 
positions that are allowed (\wrt\ $\Pi$).

\begin{example}
\label{ex_2nd}
Consider the following rewrite system, cf.\ e.g.\ \cite{Luc01b}:
\begin{equation*}
 \begin{tabular}[b]{r@{ $\rightarrow$ }l@{\;\;\;\;\;\;\;\;\;\;\;\;}r@{ $\rightarrow$ }l}
    $\mathsf{inf}(x)$ & $x : \mathsf{inf}(s(x))$ & $\mathsf{2nd}(x : (y : zs))$ & $y$
  \end{tabular}
\end{equation*}
We use one forbidden pattern $\Pi = \{\langle x : (y : z), 2.2, h \rangle\}$.
Then the term $s = 0 : s(0) : \mathsf{inf}(s(s(0)))$ is a normal form \wrt\ 
rewriting with forbidden patterns, we also say it is a $\Pi$-normal form.
Here, $x : (y : z)$ matches $s$ and the only potential redex $\mathsf{inf}(s(s(0)))$ cannot
be reduced, as it occurs at the forbidden position $2.2$ in $s$.
\end{example}

\section{Contextual Dependency Pairs}
\label{cdp}

For our approach of termination analysis of rewriting with forbidden patterns
we restrict our attention to forbidden patterns with $b$- and $h$-flags.
For brevity we call these patterns $b$- and $h$-patterns.  

We base our approach
on the well-known dependency pair (DP) framework of \cite{jar06}, which is in turn
based on dependency pairs of \cite{ag00}. 
The central observation
of the (ordinary) dependency pair approach is that given a non-terminating
rewrite system $\mathcal{R}$, there exists an infinite
reduction sequence (starting w.l.o.g.\ with a root reduction step), 
such that no redex contracted in this sequence
contains a non-terminating proper subterm. 
Such reduction sequences roughly correspond to minimal dependency pair
chains whose existence or non-existence is analyzed in the DP framework. 
For rewriting with forbidden patterns the above
observation does not hold.
\begin{example}
\label{ex_inf}
Consider the following TRS $\mathcal{R}$
\begin{equation*}
 \begin{tabular}[b]{r@{ $\rightarrow$ }l@{\;\;\;\;\;\;\;\;\;\;\;\;}r@{ $\rightarrow$ }l}
    $a$ & $f(a)$ & $f(x)$ & $g(x)$
  \end{tabular}
\end{equation*}
and an associated set of forbidden patterns
$\Pi = \{\langle f(x), 1, h \rangle\}$.
$\mathcal{R}$ is not $\Pi$-terminating:
$a \rightarrow_{\Pi} f(a) \rightarrow_{\Pi} g(a) \rightarrow_{\Pi} g(f(a)) \rightarrow_{\Pi} 
g(g(a)) \rightarrow_{\Pi} \ldots$
Note that since position $1$ is forbidden in $f(a)$, we do not have
$f(a) \rightarrow_{\Pi} f(f(a))$. Obviously, every
non-$\Pi$-terminating 
term $s$ must contain exactly one $a$. After this $a$ is reduced,
the single $a$-symbol in the contracted term is forbidden (as it occurs
in the first argument of $f$). Hence, the redex of the following reduction must
properly contain $a$. 
\end{example}

In Example \ref{ex_inf}
reductions whose redexes properly
contain non-$\Pi$-ter\-mi\-nating terms are crucial for the existence of
infinite $\Pi$-derivations.
Hence, instead of ordinary non-termination we 
focus on
a restricted form of non-$\Pi$-termination, namely
non-$\Pi$-termination in a context.

\begin{definition}[\emph{termination in a context}]
\label{def_term_cont}
Let $\mathcal{R}$ be a TRS and $\Pi$ be a set
of forbidden patterns. A term $s$ is \emph{$\Pi$-terminating in
context} $C[\Box]_p$ if $C[s]_p$ does not admit an infinite
$\Pi$-reduction sequence where each redex contracted occurs at,
below or parallel to 
$p$ and where infinitely many 
steps are at or below
$p$.
\end{definition}
We omit explicit reference to the context if 
it is clear which one is meant. 
For instance
the term $s|_p$ is $\Pi$-terminating in its context means that $s|_p$ is $\Pi$-terminating
in the context $s[\Box]_p$.

We say a term $s$ is 
\emph{minimal}
non-$\Pi$-terminating 
in a context 
$C[\Box]_q$
(w.r.t.\
a rewrite system $\mathcal{R}$ and a set of forbidden patterns $\Pi$)
if $s$ is non-$\Pi$-terminating in 
$C[\Box]_q$
and 
every proper subterm $s|_p$ of $s$
is $\Pi$-terminating in $C[s[\Box]_p]_q$.
The following lemma provides some insight into the
shape of infinite $\Pi$-reduction sequences starting from
minimal non-terminating terms.
\begin{lemma}
\label{lem_min_nonterm}
Let $\mathcal{R}$ be a TRS and let $\Pi$ be a set of forbidden $h$- and
$b$-patterns. A term $s$ that is minimal non-$\Pi$-terminating in a context
$C[\Box]_q$ admits a reduction sequence
\begin{equation*}
C[s]_q \overset{\smash{\not \leq q}}{\rightarrow}_{\Pi}^* C'[s']_q = C'[l \sigma]_q 
\overset{q}{\rightarrow}_{\Pi} C'[r \sigma]_q = C'[t]_q
\end{equation*}
such that $t$ contains a subterm $t|_p$ that is minimal
non-$\Pi$-terminating in the context $C'[t[\Box]_p]_q$.
\end{lemma}
\begin{proof}
As $s$ is non-$\Pi$-terminating in $C[\Box]_q$ there is an infinite
$\Pi$-reduction sequence starting from $C[s]_q$ such that all
reduction steps are at, below or parallel to $q$ and infinitely many
reduction steps are at or below $q$ (according to Definition
\ref{def_term_cont}). Since $s$ is minimally non-$\Pi$-terminating in
$C[\Box]_q$, eventually there must be a step at position $q$ in this
reduction sequence.
Otherwise, by the pigeonhole principle infinitely many reduction steps
would occur at or below a proper subterm of $s$ contradicting
termination 
of this subterm in its context.
Hence, we have
\begin{equation*}
C[s]_q \overset{\smash{\not \leq q}}{\rightarrow}_{\Pi}^* C'[s']_q = C'[l \sigma]_q 
\overset{q}{\rightarrow}_{\Pi} C'[r \sigma]_q = C'[t]_q
\end{equation*}
as part of our infinite $\Pi$-reduction sequence. Because of infinity
of the reduction sequence $t$ must be non-$\Pi$-terminating in
$C'[\Box]_q$. However, every term $t$ that is non-$\Pi$-terminating in
a context $C'[\Box]_q$ has a subterm $t|_p$ that is minimally
non-$\Pi$-terminating in $C'[t[\Box]_p]_q$. 
\end{proof}
Note that in contrast to
ordinary rewriting and standard minimal non-terminating terms
one can in general not assume that $p \in Pos_{\mathcal{F}}(r)$
(this effect similarly
exists in context-sensitive rewriting, cf.\ 
\cite{fsttcs06,lpar08,ic10-alarcon-gutierrez-lucas}).

\begin{example}
Consider $\mathcal{R}$ and $\Pi$ of Example \ref{ex_inf} and 
the term $f(a)$ which is minimally non-$\Pi$-terminating (in the empty context), since
position $1$ is forbidden in $f(a)$ according to $\Pi$. Now consider the reduction
$f(a) = f(x) \sigma \overset{\epsilon}{\rightarrow}_{\Pi} g(a) = g(x) \sigma$
($x \sigma = a$). The term $g(a)$ contains only one proper minimal
non-$\Pi$-terminating subterm $g(a)|_1 = a$ despite the fact that
$1 \not\in Pos_{\mathcal{F}}(g(x))$.
\end{example}

In our approach we pay attention to this phenomenon by having additional
dependency pairs to explicitly mimic the necessary extractions
of (minimally non-$\Pi$-terminating) subterms 
in DP chains (cf.\ $V_c, A_c$ and $S_c$ in Definition \ref{def_dp} below).
Technically, these rules (which we call structural dependency pairs) 
model the explicit extraction of minimal
non-$\Pi$-terminating terms on DP-chains and the introduction
of the suitable dependency pair symbol at the root of these terms.
This mechanism is similar to the way migrating variables are dealt with
in the context-sensitive dependency pair approach of \cite{lpar08}.
However, there using the concepts of ``hidden terms'' and ``function symbols
hiding positions'' it is sufficient to perform subterm extractions
out of contexts of hidden terms in right-hand sides of rewrite rules
and over arguments of functions hidden by the function. 

In the case of forbidden patterns it is necessary to use a more general
mechanism of subterm extraction, since whether a term is hidden within
the right-hand side of a rewrite rule (i.e., forbidden but might
eventually be activated) may depend on the context the right-hand side of the
rule is located in and the concrete instance of this right-hand side.
Hence, in sharp contrast to the context-sensitive dependency pair 
framework of \cite{lpar08} the structural forbidden pattern dependency pairs
associated to a TRS model subterm extractions out of arbitrary 
contexts (cf.\ $V_c, A_c$ and $S_c$ in Definition \ref{def_dp} below). 

However, we cannot disregard the contexts from which minimal non-$\Pi$-terminating
terms are extracted on DP-chains, since these contexts may contribute to
the matching of a forbidden pattern thus influencing the status of some
position in the minimal non-$\Pi$-terminating term.
In order to keep track of the subterm extractions in dependency pair
chains a context is associated to each dependency pair.
It represents the
context from which a minimal non-$\Pi$-terminating term is extracted when
the dependency pair is applied.

Informally, this amounts to an extended \emph{contextual} version of
dependency pairs 
which incorporates the full information of the given rules (especially
the complete right-hand sides) in the form of associated contexts, but
which still enables the typical DP-based reasoning enriched by
\emph{structural} DP-rules that can descend into variable
subterms of right-hand sides as well as to control where subsequent
DP-reductions are allowed to take place.

Before defining contextual dependency pairs we observe that sometimes
positions of right-hand sides are forbidden regardless of the instantiation
or location (in a context) of this right-hand side. In particular
positions forbidden by
\emph{stable} forbidden patterns have this property. We will use this
observation to reduce the number of dependency pairs that we have to consider
(cf.\ Definition \ref{def_dp} below).

\begin{definition}[stable forbidden pattern]
\label{def_stb}
Given a rewrite system $\mathcal{R}$, a forbidden pattern $\pi = (t, p, \lambda)$ is
called \emph{stable} if $t$ is linear and no rule overlaps $t$ at any position
parallel to $p$ if $\lambda = b$ and no rule overlaps $t$ at any position
parallel to or below $p$ if $\lambda = h$. 
By $Stb(\Pi)$ we denote the subset of stable forbidden patterns of $\Pi$.
\end{definition}

The crucial property of stable forbidden patterns is that the status of
positions forbidden by such patterns in some term is not altered
through reductions of $s$ at positions parallel to or below the forbidden one.

\begin{lemma}
\label{lem_stb_patterns}
Let $\mathcal{R}$ be a rewrite system and let $\pi$ be a stable pattern
matching a subterm $s|_p$ of a term $s$ (and thus forbidding some position $p.q$
in $s$). Then, if $s \overset{p'}{\rightarrow}_{\mathcal{R}, \Pi} t$ for some $p' \parallel p.q$
or $p' > p.q$, 
$p.q$ is 
also
forbidden in $t$.
\end{lemma}
\begin{proof}
Let $\pi = \langle u, o, \lambda \rangle$. First, if $\lambda = b$, then
$p'$ cannot be below $p.q$ since positions below $p.q$ are forbidden by
$\pi$. Hence, $p'$ is parallel to $p.q$. However, since $u$ is linear
and not overlapped by any rewrite rule from $\mathcal{R}$ parallel to 
$o$, $u$ matches $t|_p$ and thus $p.q$ is 
also 
forbidden in $t$.

Second, if $\lambda = h$, then $u$ is not overlapped by any rule of
$\mathcal{R}$ parallel to $o$ or below $o$. Hence, (also because $u$ is linear)
$u$ matches $t|_p$ and thus $p.q$ is forbidden in $t$.
\end{proof}

We are now ready to define the notion of contextual dependency pairs (CDPs) associated to
a rewrite system with forbidden patterns, CDP-problems and
CDP-chains.

\begin{definition}[extended \emph{contextual} dependency pairs]
\label{def_dp}
Let $(\mathcal{F}, R)$ be a TRS where the signature is partitioned
into defined symbols 
$\mathcal{D}$
and constructors 
$\mathcal{C}$. 
The set of (extended)
\emph{contextual dependency pairs (CDPs)} $CDP(\mathcal{R})$ is
given by  
$DP_c(\mathcal{R}) \uplus V_c(\mathcal{R}) \uplus A_c(\mathcal{R})
\uplus S_c(\mathcal{R})$, 
where
\small
\begin{eqnarray*}
DP_c(\mathcal{R})&=&\{l^{\#} \rightarrow r|_p^{\#}\; [c] \mid l
\rightarrow r \in R, p \in Pos^{Stb(\Pi)}_{\mathcal{D}}(r), c = r[\Box]_p \} \\ 
V_c(\mathcal{R})&=&\{l^{\#} \rightarrow T(r|_p)\; [c] \mid l \rightarrow
r \in R, r|_p = x \in Var, c = r[\Box]_p \} \\ 
A_c(\mathcal{R})&=&\{ T(f(x_1, \ldots, x_{ar(f)})) \rightarrow
f^{\#}(x_1, \ldots, x_{ar(f)})\; [\Box] \mid  
l \rightarrow r \in R, root(r|_p) = f \in \mathcal{D} \} \\
S_c(\mathcal{R})&=&\{T(f(\vec x)) \rightarrow  
T(x_i) [f(\vec x)[\Box]_i] \mid \vec x = x_1, \ldots, x_{ar(f)}, l
\rightarrow r \in R, root(r|_p) = f, i \in \{1, \dots, ar(f)\} \}\,.
\end{eqnarray*}
\normalsize
Here, $T$ is a new auxiliary function symbol (the \emph{token} symbol
for ``shifting attention''). We call $V_c(\mathcal{R})$ \emph{variable descent}
CDPs, $S_c(\mathcal{R})$ \emph{shift} CDPs\ and $A_c(\mathcal{R})$
\emph{activation} CDPs. 
\end{definition}

\begin{example}
\label{ex_inf_dp}
Consider the TRS $\mathcal{R}$ of Example \ref{ex_inf}. Here,
$CDP(\mathcal{R})$ consists of:
\begin{equation*}
 \begin{tabular}[b]{r@{ $\rightarrow$ }l@{\;\;\;\;\;\;}r@{ $\rightarrow$ }l@{\;\;\;\;\;\;}r@{ $\rightarrow$ }l@{\;\;\;\;\;\;}r@{ $\rightarrow$ }l}
    $a^{\#}$ & $a^{\#} [f(\Box)]$ & $a^{\#}$ & $f^{\#}(a) [\Box]$ & $f^{\#}(x)$ & $T(x) [g(\Box)]$ & $T(a)$ & $a^{\#} [\Box]$ \\
    $T(f(x))$ & $f^{\#}(x) [\Box]$ & $T(g(x))$ & $g^{\#}(x) [\Box]$ & $T(f(x))$ & $T(x) [f(\Box)]$ & $T(g(x))$ & $T(x) [g(\Box)]$
  \end{tabular}
\end{equation*}
\end{example}

\emph{Contextual} rules of the shape $l \rightarrow r \; [c]$ can be
interpreted as $l \rightarrow c[r]$ (provided that $Var(c[r]) \subseteq Var(l)$)
when used
as rewrite rules. Slightly abusing notation, for a set $\mathcal{P}$
of such contextual rewrite rules (i.e.~a \emph{contextual TRS}) we
denote by $\rightarrow_\mathcal{P}$ 
the corresponding induced ordinary rewrite relation.
Based on our notion of contextual dependency pairs, we now define
forbidden pattern contextual dependency pair problems (FP-CDP problems) and
forbidden pattern contextual dependency pair chains (FP-CDP chains). Proving the
absence of infinite FP-CDP chains is the main goal of a CDP based
attempt to prove $\Pi$-termination (cf.\ Theorem \ref{thm_cdp} below). 

\begin{definition}[forbidden pattern CDP problem]
A forbidden pattern CDP problem (FP-CDP
problem or just CDP problem) is a quadruple
$(\mathcal{P}, \mathcal{R}, \Pi, T)$ where $\mathcal{P}$ is a
contextual TRS, $\mathcal{R} = (\mathcal{F}, R)$ is a TRS, $\Pi$ is a set of forbidden
patterns over $\mathcal{F}$ and $T$ is a designated
function symbol with $T \not \in \mathcal{F}$
that occurs only at the root position of 
left- and right-hand sides 
of rules
in $\mathcal{P}$
(but not, for example, in contexts).
\end{definition}

\begin{definition}[forbidden pattern CDP chain]
\label{def_chain}
Let $(\mathcal{P}, \mathcal{R}, \Pi, T)$ be a CDP problem where $\mathcal{R} = (\mathcal{F}, R)$. 
The sequence $S \colon s_1 \rightarrow t_1\;[c_1[\Box]_{p_1}], s_2 \rightarrow t_2\;[c_2[\Box]_{p_2}], \dots$
is a $(\mathcal{P}, \mathcal{R}, \Pi, T)$-CDP chain (we also say FP-CDP chain or just CDP chain if
the CDP problem is clear from the context) if
\begin{itemize}
\item there exists a substitution $\sigma: Var \rightarrow \mathcal{T}(\mathcal{F}, V)$, such that
\begin{eqnarray*}
s_1 \sigma & \rightarrow_{\mathcal{P}} & c_1[t_1 \sigma]_{p_1} =
  c_1'[t_1 \sigma]_{p_1'} \\ 
\overset{\smash{\not \leq p_1'}}{\rightarrow}_{\mathcal{R}}^*
c_1''[s_2 \sigma]_{p_1'} & \rightarrow_{\mathcal{P}} & c_1''[c_2[t_2
  \sigma]_{p_2}]_{p_1'} = c_2'[t_2 \sigma]_{p_2'}\\ 
\overset{\smash{\not \leq p_2'}}{\rightarrow}_{\mathcal{R}}^* 
c_2''[s_3 \sigma]_{p_2'} & \rightarrow_{\mathcal{P}} & c_2''[c_3[t_3
\sigma]_{p_3}]_{p_2'} = c_3'[t_3 \sigma]_{p_3'} \;\dots 
\end{eqnarray*} 
where 
$c_i' = c_{i-1}''[c_i]$ and $p_i' = p_{i-1}' p_i$ for all $1 \leq i$
($s_1\sigma = c_0''[s_1\sigma]_{p_0'}$ with $p_0' = \epsilon$, $c_0'' = \Box$),
\item the $\mathcal{R}$-reduction 
$c_i'[t_i \sigma]_{p_i'}
\overset{\smash{\not \leq p_i'}}{\rightarrow}_{\mathcal{R}}^*
c_i''[s_{i+1} \sigma]_{p_i'}$ is empty (i.e.,  
$c_i'[t_i \sigma]_{p_i'} = c_i''[s_{i+1} \sigma]_{p_i'}$) 
whenever $root(t_i) = T$ 
(i.e., the token symbol), and
\item
for each single reduction $s \overset{q}{\rightarrow}_{\mathcal{P}} t$ or $s
\overset{q}{\rightarrow}_{\mathcal{R}} t$  
in this reduction sequence position $q$ is allowed in $erase(s)$ according to $\Pi$.
Here $erase(s)$ is obtained 
from $s$
by replacing all marked dependency pair symbols $f^{\#}$ by their
unmarked versions $f$ and by replacing terms $T(s')$ by $s'$.%
\footnote{
Note that this definition makes sense since whenever a $T$ occurs in $s$, then
$q$ is not below the occurrence of $T$.
Moreover, this definition of $erase$ is formally not compatible with the DP framework, 
since it is based on the 
correspondence of marked dependency pair symbols to the original
  function symbols from which they originated. This correspondence
  might not 
exist in arbitrary CDP problems. However, to restore full modularity the $erase$ function
could be
made
part of the notion of
CDP problem. We refrain from doing so for notational
simplicity.
}
\end{itemize} 
Moreover, $S$ is minimal if
for every $i \geq 0$ every subterm of 
$c_{i}'[t_i\sigma]_{p_{i}'}$
at position $q > p_{i}'$ is
$\Pi$-terminating in its context (w.r.t.\ $\mathcal{R}$).
\end{definition}

\begin{example}
Consider the TRS $\mathcal{R}$ and $\Pi$ from Example \ref{ex_inf} ($CDP(\mathcal{R})$
is given in Example \ref{ex_inf_dp}) and the corresponding FP-CDP 
$P = (CDP(\mathcal{R}), \mathcal{R}, \Pi, T)$. $P$ admits an infinite CDP chain:
\begin{equation*}
a^{\#} \rightarrow f^{\#}(a)\; [\Box], f^{\#}(x) \rightarrow T(x)\;
[g(\Box)], T(a) \rightarrow a^{\#}\; [\Box], \ldots 
\end{equation*}  
corresponding to 
\begin{equation*}
a^{\#} \rightarrow_{DP_c(\mathcal{R})} f^{\#}(a)
       \rightarrow_{V_c(\mathcal{R})}  g(T(a))
       \rightarrow_{A_c(\mathcal{R})}  g(a^{\#}) 
       \rightarrow_{DP_c(\mathcal{R})} g(f^{\#}(a)) \ldots
\end{equation*}  
\end{example}

We say a CDP problem is \emph{finite} if it does not admit an 
infinite minimal CDP chain.
Indeed, the existence of infinite $(CDP(\mathcal{R}), \mathcal{R}, \Pi, T)$-chains
coincides with non-$\Pi$-termination of $\mathcal{R}$. Before proving
this we provide a lemma stating that forbidden rewrite steps can
be extracted out of contexts.

\begin{lemma}[extraction lemma]
\label{lem_extraction}
If $C[s]_p \overset{\geq p}{\rightarrow}_{\Pi} C[t]_p$, then 
$s \rightarrow_{\Pi} t$.
\end{lemma}
\begin{proof}
Immediate by the definition of rewriting with forbidden patterns.
\end{proof}

\begin{theorem}
\label{thm_cdp}
Let $\mathcal{R}$ be a TRS with an associated set of forbidden
patterns $\Pi$. $\mathcal{R}$ is $\Pi$-terminating if
and only if the FP-CDP problem $(CDP(\mathcal{R}), \mathcal{R},
\Pi, T)$ 
is finite. 
\end{theorem}
\begin{proof}
\textsc{if: } 
Let $\mathcal{R}$ be non-$\Pi$-terminating. 
According to Lemma \ref{lem_min_nonterm}, there exist terms $s,
  s_i, t_i$, $t_i'$ and an infinite reduction sequence $S$ of shape 
\begin{eqnarray*}
S\colon s & \overset{\smash{> \epsilon^{\textcolor{white}{'}}}}{\rightarrow}_{\Pi}^* & 
t_1 \overset{\smash{\epsilon}}{\rightarrow}_{\Pi}
s_1 = C_2'[t_2']_{p_2} \\
&\overset{\smash{\not\leq p_2}}{\rightarrow}_{\Pi}^* & C_2[t_2]_{p_2}
\overset{p_2}{\rightarrow}_{\Pi} C_2[s_2]_{p_2} = C_3'[t_3']_{p_3} \\
&\overset{\smash{\not\leq p_3}}{\rightarrow}_{\Pi}^* & C_3[t_3]_{p_3}
\overset{p_3}{\rightarrow}_{\Pi} C_3[s_3]_{p_3} = C_4'[t_4']_{p_4} \dots
\end{eqnarray*}
such that
$p_i \leq p_{i+1}$, $t_i'$ is minimally non-$\Pi$-terminating in
  $C_i'[\Box]_{p_i}$ for 
all $i \geq 1$, and every proper subterm of $s$ is $\Pi$-terminating
(regardless of the context, hence $s$ is also minimal non-$\Pi$-terminating in $\Box$). 
Here, $p_1 = \epsilon$, $C_1 = C_1' = \Box$ and $t_1' = t_1$.

We are going to construct 
an infinite $(CDP(\mathcal{R}), \mathcal{R}, \Pi, T)$-chain $T$ by
associating a (sequence of) dependency pair(s) to each $C_i[t_i]_{p_i}
\overset{p_i}{\rightarrow}_{\Pi} C_i[s_i]_{p_i}$ step.
Consider one of these reduction steps 
$C_i[t_i]_{p_i}
\overset{p_i}{\rightarrow}_{\Pi, l \rightarrow r} C_i[s_i]_{p_i} = C_{i+1}'[t_{i+1}']_{p_{i+1}}$.
Let $p_i.q = p_{i+1}$; we distinguish 2 cases:

First, if $q \in Pos_{\mathcal{F}}(r)$, then the dependency pair
$l^{\#} \rightarrow r|_q^{\#} [c] \in DP_c(\mathcal{R})$ is used. Note that
$root(r|_q) \in \mathcal{D}$, as $t_{i+1}'$ is minimally non-$\Pi$-terminating.
Moreover, $q \in Pos^{Stb(\Pi)}$ since otherwise $p_{i+1}$ would be
forbidden by a stable forbidden pattern in $C_{i+1}'[t_{i+1}']_{p_{i+1}}$
and thus also in every term obtained from $C_{i+1}'[t_{i+1}']_{p_{i+1}}$ through reduction parallel to
or below
$p_{i+1}$, due to Lemma \ref{lem_stb_patterns}. Hence, there could not be a further step at position
$p_{i+1}$ contradicting the existence of a reduction chain of the above shape.
Finally, we also have $C_i[c] = C_{i+1}'$ by Definition \ref{def_dp}.

Second, if $q \not\in Pos_{\mathcal{F}}(r)$, let $q' \leq q$ be
the unique variable position of $r$ that is above $q$. Now we construct
a sequence of dependency pairs starting with 
$l^{\#} \rightarrow T(x) [c] \in V(\mathcal{R})$ where
$c = r[\Box]_{q'}$. By using this 
dependency pair we ``introduce'' the token symbol $T$ at position
$q'$ in $s_i$. The goal now is to shift it to position $q$.

In the following we say that
a function symbol 
$f$
is a \emph{shift symbol} if there exist
dependency pairs $T(f(\vec x)) \rightarrow T(x_i) [f(\vec x)[\Box]_i]$
for all $i \in \{1, \ldots, ar(f)\}$. 
Assume $q' \not= q$ (say
$q'.i.o = q$) and let $root(s_i|_{q'}) = f$. 

If $f$ is not a shift symbol, then $f$ does not occur in the
right-hand side of a rewrite rule at all (according to Definition \ref{def_dp}). 
However, if $f$ does not occur in the
right-hand side of any rule of $\mathcal{R}$, $s_i|_{q'}$ must be the
descendant of some proper subterm of $s$. However, $s_i|_{q'}$ is
non-$\Pi$-terminating since it contains $t_{i+1}'$ which is not
$\Pi$-terminating in its context.
Thus $s_i|_{q'}$
cannot be a successor
of such a proper subterm of $s$, since these subterms were assumed
to be $\Pi$-terminating (in any context) (cf.\ also Lemma \ref{lem_extraction}).  

Hence, $f$ is a shift symbol and thus there is a dependency pair
$T(f(x_1, \ldots, x_{ar(f)})) \rightarrow T(x_i) [c] \in S_c(\mathcal{R})$
where 
$f(x_1, \ldots, x_{i-1}, \Box, x_{i+1}, \ldots, x_{ar(f)}) = c$.
By adding this dependency pair we 
shift
the token symbol to position
$q'.i$ in $s_i$ (more precisely with the addition of the shift dependency pair
we are now considering a term $s_i'$ with $erase(s_i') = erase(s_i)$ where the
unique occurrence of the token symbol is at position $q'.i$). 
If $q'.i \not= q$ 
we add more dependency
pairs from $S_c(\mathcal{R})$ to shift the token symbol to
$q'.i.i'$, $q'.i.i'.i''$, $\ldots$, until the token
is finally shifted to $q$.

Finally, we add the 
activation
dependency pair 
$T(g(\vec x)) \rightarrow g^{\#}(\vec x) [\Box] \in
A_c(\mathcal{R})$, 
where
$g = root(s_i|_q)$. 
Note that, as for the shift dependency pairs, here
$g$ must occur in the right-hand side of some rewrite rule,
since otherwise $s_i|_q$ would be the descendent of some
proper subterm of $s$ which contradicts non-$\Pi$-termination
of $t_{i+1}'$. 

Moreover, since $s_i|_q$ is minimally non-$\Pi$-terminating, we have 
$root(s_i|_q) \in \mathcal{D}$.

It is easy to see that the infinite sequence of dependency pairs 
$T$
obtained by this
construction actually forms an infinite DP chain, where
$\sigma$ is given by the substitutions used in
the $C_i[t_i]_{p_i}
\overset{p_i}{\rightarrow}_{\Pi} C_i[s_i]_{p_i}$ steps
of $S$ (note that we consider CDPs in chains to be variable
disjoint). The fact that we actually have a valid CDP chain
is a direct consequence of the particular choice of $S$. 

\textsc{only if: } If there exists an infinite
CDP-chain we obtain an infinite $\mathcal{R}$-reduction
by considering the 
($CDP(\mathcal{R}) \cup \mathcal{R}$)-reduction
of Definition \ref{def_chain}. Then by
applying $erase$ to every term in this chain, we get
that every single ($CDP(\mathcal{R}) \cup \mathcal{R}$)-step can be
simulated by 0 or 1 $\rightarrow_{\mathcal{R}}$-reduction
steps. Here the simulating reduction is 
empty only if a $CDP(\mathcal{R})$-step
with rules 
from $S_c(\mathcal{R})$ or $A_c(\mathcal{R})$ occurs.
However,
it is easy to see that no infinite
$CDP(\mathcal{R}) \cup \mathcal{R}$-reduction sequence 
can use only these rules, hence
the simulating $\mathcal{R}$-reduction is infinite as well.
\end{proof}

Now, following the dependency pair framework of \cite{jar06}
we define CDP processors as functions mapping CDP problems
to sets of CDP problems.

It is easy to observe that each FP-CDP chain
w.r.t.\ a FP-CDP problem 
$(\mathcal{P}, \mathcal{R}, \Pi, T)$ is also an ordinary
(though not 
necessarily 
minimal) DP\ chain w.r.t.\ $(\mathcal{P}, \mathcal{R})$
(when disregarding the contexts of DPs).
Hence, in some cases processors that are sound in the ordinary DP framework
of \cite{jar06} and do not rely on minimality can be adapted to work also in the
forbidden pattern contextual extension of the
DP framework. One example of such a processor is the reduction pair
processors not using usable rules (\cite{jar06}). Another important
example is the dependency graph processor. Both processors have been
used in our experiments. 
In both cases, given a CDP problem
$(\mathcal{P}, \mathcal{R}, \Pi, T)$, the processors are applied to
the ordinary DP problem $(\mathcal{P}', \mathcal{R})$, where $\mathcal{P}'$
is obtained from $\mathcal{P}$ 
by stripping off the contexts of the contextual rules.

\section{A Specific CDP Processor}
\label{scp}

In the following we develop a method to prove the absence
of minimal CDP chains by inspecting the contexts of dependency pairs.
To this end we consider the nested contexts of 
consecutive dependency pairs of candidates for infinite DP chains.
Then, if for such a candidate in the obtained nested contexts of 
consecutive dependency pairs the 
unique
box position is forbidden (by certain forbidden patterns), 
the candidate chain is not a proper FP-CDP chain. 
A CDP processor could then soundly delete a CDP $s \rightarrow t [c]$ 
from a CDP-problem if no candidate chain containing
$s \rightarrow t [c]$ is a proper FP-CDP chain 
(provided that the set of candidates is complete).

\begin{example}
\label{ex_analyzing_contexts}
Consider a CDP problem $(\mathcal{P}, \mathcal{R}, \Pi, T)$ where
\begin{equation*}
\begin{tabular}{r@{ $\;\;=\;\;$ }l@{$\;\;\;\;\;\;\;\;\;$}r@{ $\;\;=\;\;$ }l}
$\mathcal{P}$ & $\{a^{\#} \rightarrow a^{\#} [f(\Box)]\}$ &
$\mathcal{R}$ & $\{a \rightarrow f(a)\}$ \\
$\Pi$ & $\{\langle f(f(f(x))), 1.1, b \rangle\}$.
\end{tabular}
\end{equation*}
If there were an infinite FP-CDP chain w.r.t.\ this CDP problem, then it
would consist of an infinite sequence of the only CDP $a^{\#} \rightarrow a^{\#} [f(\Box)]$.
Hence, this sequence is the only candidate for an FP-CDP chain. Now
considering 
the contexts occurring in this CDP chain candidate we get $f(f(...(\Box)...))$
(for any sufficiently large finite subsequence). However, in this term
the box position is forbidden by $\Pi$. Hence the CDP chain candidate
is not a proper FP-CDP chain and since it was the only candidate we conclude
finiteness of the CDP problem.
\end{example}

Unfortunately, there are two major problems with this approach. First, in order
to obtain a sound CDP processor one would have to consider
candidates for CDP chains in a complete way. 
Second,
according to Definition \ref{def_chain} contexts are not constant but
may be modified at positions parallel to the box position in FP-CDP chains.

We will deal with the second problem first, starting with the
observation 
that the (nested)
contexts are stable modulo reductions parallel to the position of
the hole, i.e.\ they are altered only through reductions parallel
to the hole position. Hence, if forbidden patterns oblivious to
this kind of parallel reductions forbid the hole position in
such a context, the corresponding sequence of dependency pairs
cannot form an FP-CDP chain according to Definition \ref{def_chain}. 
We characterize (or rather approximate)
these patterns by the definition of the subset $\Pi_{orth}$ of $\Pi$.
The name $\Pi_{orth}$ expresses that these forbidden patterns
are orthogonal to a given rewrite system $\mathcal{R}$ in that
they are not overlapped by rules of $\mathcal{R}$.

\begin{definition}[$\Pi_{orth}$]
Let $\mathcal{R}$ be a TRS and $\Pi$ be a set of corresponding
forbidden patterns. The set $\Pi_{orth} \subseteq \Pi$ consists of
those forbidden patterns $\langle t, p, \lambda\rangle$
where $\lambda \in \{h, b\}$, $t$ is linear and 
not overlapped by any rule of $\mathcal{R}$
at any position that is parallel to or below $p$.
\end{definition}

The following lemma is the key result for analyzing nested contexts
of CDP chain candidates. It states that whenever the box position $q$
of a nested context corresponding to a CDP chain candidate (after substituting the
right-hand side of the last CDP) is forbidden, then this position is
also forbidden in every other term obtained from the nested context by
rewriting at positions parallel to $q$.

\begin{lemma}
\label{lem_contexts}
Let $(\mathcal{P}, \mathcal{R}, \Pi, T)$ be a CDP problem and let
$s_1 \rightarrow t_1 [c_1], \ldots, s_n \rightarrow t_n [c_n]$ be a sequence of CDPs.
If position $p_1.\cdots.p_n$ is forbidden in the term 
$c_1[c_2[\ldots c_n[erase(t_n)]_{p_n}\ldots]_{p_2}]_{p_1}$ by a forbidden pattern
from $\Pi_{orth}$, then the same position is
forbidden in 
$c_1'[c_2'[\ldots c_n'[t_n']_{p_n}\ldots]_{p_2}]_{p_1}$
where $c_i \rightarrow_{\mathcal{R}}^* c_i'$ with reductions parallel to $p_i$ for
all $1 \leq i \leq n$ and $erase(t_n) \overset{> \epsilon}{\rightarrow}_{\mathcal{R}} t_n'$.
\end{lemma}
\begin{proof}
For brevity let 
\begin{equation*}
c[t_n]_{q} = c_1[c_2[\ldots c_n[erase(t_n)]_{p_n}\ldots]_{p_2}]_{p_1},
\end{equation*}
where $q = p_1.p_2.\cdots.p_n$ and let 
\begin{equation*}
c'[t_n']_q \mbox{ be some } c_1'[c_2'[\ldots c_n'[t_n']_{p_n}\ldots]_{p_2}]_{p_1}
\end{equation*}
where $c_i \rightarrow_{\mathcal{R}}^* c_i'$ with reductions parallel to $p_i$ for
all $1 \leq i \leq n$ and $erase(t_n) \overset{> \epsilon}{\rightarrow}_{\mathcal{R}} t_n'$.

Assume the forbidden pattern $\langle t, o, \lambda \rangle$ forbidding the reduction of 
$c[t_n]_q$ at position
$q$ matches the term at position $q' < q$ and assume moreover
that the same pattern does not match $c'[t_n]_q$. Since we consider
plain $\mathcal{R}$-reduction and not forbidden pattern reduction we
have $S\colon c[t_n]_q \rightarrow_{\mathcal{R}}^* 
c'[t_n']_q$ with reductions parallel to or strictly below $q$. 
Since $t$ does not match
$c'[t_n']_q|_{q'}$ and is linear, there must be some reduction at a position $q'.q''$
where $q'' \in Pos_{\mathcal{F}}(t)$ and $q''$ is either parallel to or below $o$. Hence, $t$ is overlapped by some
rule of $\mathcal{R}$ at some position parallel to or below $o$,
and we get a contradiction to $\langle t, o, \lambda \rangle \in \Pi_{orth}$.
\end{proof}

Lemma \ref{lem_contexts} establishes that for a DP chain candidate
it suffices to consider the nested contexts unmodified as long
as one only considers patterns from $\Pi_{orth}$ to check whether
the nested contexts forbid their hole position 
implying that 
the 
candidate chain is not an actual chain.

Regarding the second problem of considering a complete set of
CDP chain candidates, we present a simple solution based on the idea
of 
taking into account only
all possible candidates for CDP chains of a bounded
length. This approach ultimately leads to the definition of the
\emph{simple context processor}.

However, as indicated in \cite{wst10} more clever ways of
handling this problem may yield even better results regarding power
in termination analysis. In this work we still resort to the
simpler approach, as it enables us to perform on-the-fly synthesis
of forbidden patterns as discussed in Section \ref{synth} below.

The idea of the \emph{simple context processor} is to consider only
CDP chain candidates of a bounded length $n$. Assuming a finite set of
CDPs, there are only finitely many possible sequences of CDPs of this
length. Then, if none of these sequences containing a certain CDP
$s \rightarrow t [c]$ is an FP-CDP chain 
(which then
cannot be part of an
infinite FP-CDP chain, cf.\ Lemma \ref{lem_subsequence_chain} below) it is sound to delete
$s \rightarrow t [c]$ from the given CDP problem.

The following lemma establishes that every finite subsequence of CDPs
forming an FP-CDP chain form an FP-CDP chain in turn. 

\begin{lemma}
\label{lem_subsequence_chain}
Let $(\mathcal{P}, \mathcal{R}, \Pi, T)$ be a CDP problem and
$\alpha_1, \alpha_2, \ldots$ be an FP-CDP chain where $\alpha_i \in \mathcal{P}$
for all $i \geq 1$. 
Then $\alpha_{m}, \alpha_{m+1}, \ldots, \alpha_{m+n}$
as well as $\alpha_m, \alpha_{m+1}, \ldots$
are FP-CDP chains for all $m, n \geq 1$.
\end{lemma} 

\begin{proof}
We consider the original CDP sequence $\alpha_1, \alpha_2, \ldots$
and write $\alpha_1 = s_1 \rightarrow t_1 [c_1], \alpha_2 =
s_2 \rightarrow t_2 [c_2], \ldots$. Since this CDP sequence is an
FP-CDP chain we have
\begin{eqnarray*}
s_1 \sigma & \rightarrow_{\mathcal{P}} & c_1[t_1 \sigma]_{p_1} =
  c_1'[t_1 \sigma]_{p_1'} \\ 
\overset{\smash{\not \leq p_1'}}{\rightarrow}_{\mathcal{R}}^*
c_1''[s_2 \sigma]_{p_1'} & \rightarrow_{\mathcal{P}} & c_1''[c_2[t_2
  \sigma]_{p_2}]_{p_1'} = c_2'[t_2 \sigma]_{p_2'}\\ 
& \ldots & \\
\overset{\smash{\not \leq p_{m-1}'}}{\rightarrow}_{\mathcal{R}}^* 
c_{m-1}''[s_m \sigma]_{p_{m-1}'} & \rightarrow_{\mathcal{P}} & c_{m-1}''[c_m[t_m
\sigma]_{p_m}]_{p_{m-1}'} = c_m'[t_m \sigma]_{p_m'} \\
\overset{\smash{\not \leq p_m'}}{\rightarrow}_{\mathcal{R}}^* 
c_{m}''[s_{m+1} \sigma]_{p_{m}'} & \rightarrow_{\mathcal{P}} & c_{m}''[c_{m+1}[t_{m+1}
\sigma]_{p_{m+1}}]_{p_{m}'} = c_{m+1}'[t_{m+1} \sigma]_{p_{m+1}'} \\
& \ldots & \\
\overset{\smash{\not \leq p_{m+n-1}'}}{\rightarrow}_{\mathcal{R}}^* 
c_{m+n-1}''[s_{m+n} \sigma]_{p_{m+n-1}'} & \rightarrow_{\mathcal{P}} & c_{m+n-1}''[c_{m+n}[t_{m+n}
\sigma]_{p_{m+n}}]_{p_{m+n-1}'} \\
& \ldots & 
\end{eqnarray*} 
for some substitution $\sigma$ according to Definition \ref{def_chain}. However,
according to Lemma \ref{lem_extraction} we also have
\begin{eqnarray*}
s_m \sigma & \rightarrow_{\mathcal{P}} & c_m[t_m
\sigma]_{p_m} = \tilde{c}_m'[t_m \sigma]_{\tilde{p}_m'} \\
\overset{\smash{\not \leq \tilde{p}_m'}}{\rightarrow}_{\mathcal{R}}^* 
\tilde{c}_{m}''[s_{m+1} \sigma]_{\tilde{p}_{m}'} & \rightarrow_{\mathcal{P}} & \tilde{c}_{m}''[c_{m+1}[t_{m+1}
\sigma]_{p_{m+1}}]_{\tilde{p}_{m}'} = \tilde{c}_{m+1}'[t_{m+1} \sigma]_{\tilde{p}_{m+1}'} \\
& \ldots & \\
\overset{\smash{\not \leq \tilde{p}_{m+n-1}'}}{\rightarrow}_{\mathcal{R}}^* 
\tilde{c}_{m+n-1}''[s_{m+n} \sigma]_{\tilde{p}_{m+n-1}'} & \rightarrow_{\mathcal{P}} & \tilde{c}_{m+n-1}''[c_{m+n}[t_{m+n}
\sigma]_{p_{m+n}}]_{\tilde{p}_{m+n-1}'} \\
& \ldots &
\end{eqnarray*} 
where $\tilde{c}_i'$ (resp.\ $\tilde{c}_i''$) is obtained from $c_i'$ 
(resp.\ $c_i''$) through extraction, i.e.\ $\tilde{c}_i' = c_i'|_o$
(resp.\ $\tilde{c}_i'' = c_i''|_o$) for some position $o$ for all
$i \in \{m, \ldots, m+n, \ldots\}$.
Hence, $\alpha_{m}, \alpha_{m+1}, \ldots, \alpha_{m+n}$ resp.\
$\alpha_{m}, \alpha_{m+1}, \ldots$ are proper
FP-CDP chains.
\end{proof}

Using Lemma \ref{lem_subsequence_chain} we get that if no sequence of
CDPs of length $n$ involving a certain CDP $\alpha$ is a proper FP-CDP chain,
no infinite FP-CDP chain involves $\alpha$ and hence $\alpha$ can be soundly
deleted.
Thus, by additionally using Lemma \ref{lem_contexts} we can define
an effective CDP processor, the \emph{simple context processor}.

\begin{definition}[Simple context processor]
\label{def_scp}
Let $Prob = (\{s \rightarrow t [c[\Box]_p]\} \uplus \mathcal{P}, \mathcal{R}, \Pi, T)$ be a CDP problem. Given a
bound $n$ the simple context processor ($SCP_n$) returns 
\begin{itemize}
\item $\{(\mathcal{P}, \mathcal{R}, \Pi, T)\}$ if for every sequence
of CDPs 
\begin{equation*}
s \rightarrow t [c[\Box]_p], s_2 \rightarrow t_2 [c_2[\Box]_{p_2}], 
\ldots, s_{n} \rightarrow t_n [c_n[\Box]_{p_n}]
\end{equation*}
position $p.p_2.\cdots.p_n$ is forbidden in the term
$c[c_2[\ldots c_n[erase(t_n)]_{p_n}\ldots]_{p_2}]_p$ by a forbidden pattern
of $\Pi_{orth}$, and
\item $\{Prob\}$ otherwise.
\end{itemize}
\end{definition} 

\begin{theorem}
\label{thm_scp_sound}
The CDP processor $SCP_n$ is sound and complete for every $n > 1$.
\end{theorem}
\begin{proof}
Completeness of the processor is trivial since either one CDP is deleted
or the problem is returned unmodified. In either case infinity of the
returned problem implies infinity of the original one.

Regarding soundness assume towards a contradiction that $Prob$ is infinite
while $SCP_n(Prob)$ is finite (i.e.\ the single problem contained in
the set of returned problems). If $SCP_n(Prob) = \{Prob\}$ soundness is
trivial. Otherwise, let $Prob = (\{s \rightarrow t [c[\Box]_p]\} \uplus \mathcal{P}, \mathcal{R}, \Pi, T)$
and $SCP_n(Prob) = \{(\mathcal{P}, \mathcal{R}, \Pi, T)\}$, i.e.\ the
CDP $s \rightarrow t [c[\Box]_p]$ has been deleted by the processor.
Since $Prob$ is infinite there exists an infinite FP-CDP chain
$S\colon \alpha_1, \alpha_2, \ldots$ with $\alpha_i \in \{s \rightarrow t [c[\Box]_p]\} \uplus \mathcal{P}$
for all $i \geq 1$. Moreover, $s \rightarrow t [c[\Box]_p]$ occurs infinitely
often in $S$, since otherwise there would exist an infinite FP-CDP
chain without $s \rightarrow t [c[\Box]_p]$ starting after the last
occurrence of $s \rightarrow t [c[\Box]_p]$ 
in $S$ (using Lemma \ref{lem_subsequence_chain}), thus contradicting
finiteness of $(\mathcal{P}, \mathcal{R}, \Pi, T)$.

Now consider a subsequence of length $n+1$ of $S$ starting at an occurrence
of $s \rightarrow t [c[\Box]_p]$, i.e.\ $\alpha_m, \alpha_{m+1}, \ldots, \alpha_{m+n+1}$.
According to Definition \ref{def_scp}, since $s \rightarrow t [c[\Box]_p]$ has
been deleted, the position $p.p_{m+1}.\cdots.p_{m+n+1}$ is forbidden in the term
$c[c_{m+1}[\ldots c_{m+n}[erase(t_{m+n})]_{p_{m+n}}\ldots]_{p_{m+1}}]_p$ 
by a forbidden pattern of $\Pi_{orth}$
where
$c_i[\Box]_{p_i}$ is the context associated to the CDP $\alpha_i$ for all $i \geq 1$ and
$t_{m+n}$ is the right-hand side of the CDP $\alpha_{m+n}$.  

Using Lemma \ref{lem_contexts} we obtain that the same position is $\Pi_{orth}$-forbidden in every
term obtained from $c[c_{m+1}[\ldots c_{m+n}[erase(t_{m+n})]_{p_{m+n}}\ldots]_{p_{m+1}}]_p$
by reduction parallel to or below $p.p_{m+1}.\cdots.p_{m+n+1}$. Thus, 
there cannot be a subsequent CDP step at this position and hence
$\alpha_m, \alpha_{m+1}, \ldots, \alpha_{m+n+1}$ is not a proper FP-CDP chain.
However, by Lemma \ref{lem_subsequence_chain} this implies that $S$ is not a proper
FP-CDP chain and we get a contradiction.
\end{proof}

\begin{example}
Consider the CDP problem of Example \ref{ex_analyzing_contexts} and
let $n = 3$. The only candidate CDP sequence of length $3$ is
$\alpha, \alpha, \alpha$ where $\alpha = a^{\#} \rightarrow a^{\#} [f(\Box)]$.
The nested context corresponding to this CDP sequence is
$f(f(f(\Box)))$, the relevant position is $1.1.1$ and $f(f(f(erase(a^{\#})))$ is $f(f(f(a)))$.
In this example $\Pi_{orth} = \Pi$ and thus we observe that position
$1.1.1$ is forbidden in the term $f(f(f(a)))$. According to Theorem
\ref{thm_scp_sound} it is sound to delete $\alpha$, thus leaving us with
an empty set of CDPs. Hence, we conclude finiteness of the original CDP problem.
\end{example}

Definition \ref{def_scp} requires to consider all sequences of CDPs of a given
length $n$ as CDP chain candidates. However, in practice it is not desirable to 
consider all $n$-tuples of CDPs since the number of these tuples combinatorially
explodes. To counter this problem the sequences of CDPs that need to be considered
can be obtained from existing DP graph approximations (cf.\ e.g.\ \cite{jar06,frocos05,cade09}). 
By the definition of the dependency graph, every DP-chain corresponds to
a path in this graph and also every FP-CDP chain corresponds to a
path in the DP graph and thus also in every (over-)approximation of this
graph.

\begin{example}
Consider the TRS $\mathcal{R}$ from Example \ref{ex_2nd} and one forbidden
pattern $\Pi = \{\langle x : (y : zs), \epsilon, b \rangle \}$. We have
$CDP(\mathcal{R}) = $
\begin{equation} \nonumber	
  \begin{tabular}[b]{r@{ $\rightarrow$ }l@{\;\;\;\;\;\;}r@{ $\rightarrow$ }l}
    $\{\alpha_1: inf^{\#}(x)$ & $inf^{\#}(s(x)) \; [x : \Box]$ & $\alpha_2: inf^{\#}$ & $T(x) \; [x : inf(s(\Box))]$ \\
    $\alpha_3: inf^{\#}$ & $T(x) \; [\Box : inf(s(x))]$ & $\alpha_4: 2nd^{\#}(x : y : zs)$ & $T(y) \; [\Box]$ \\
    $\alpha_5: T(inf(x))$ & $T(x) \; [inf(\Box)]$ & $\alpha_6: T(inf(x))$ & $inf^{\#}(x) \; [\Box]\}$
  \end{tabular}
\end{equation}
Now we apply the simple context processor to the CDP problem
$(CDP(\mathcal{R}), \mathcal{R}, \Pi, T)$ with a bound of $n = 3$ and
considering the CDP $\alpha_1$. By computing some DP graph approximation
one observes that all DP chain candidates of length $3$ starting with
the CDP $\alpha_1$ are the following.
\begin{eqnarray*}
\alpha_1, \alpha_2, \alpha_5 & \mbox{ with corresponding context: } & x : (x' : inf(s(inf(\Box)))) \\
\alpha_1, \alpha_2, \alpha_6 & \mbox{ with corresponding context: } & x : (x' : inf(s(\Box))) \\
\alpha_1, \alpha_3, \alpha_5 & \mbox{ with corresponding context: } & x : (inf(\Box) : inf(s(x'))) \\
\alpha_1, \alpha_3, \alpha_6 & \mbox{ with corresponding context: } & x : (\Box : inf(s(x'))) \\
\alpha_1, \alpha_1, \alpha_2 & \mbox{ with corresponding context: } & x : (x' : (x'' : inf(s(\Box)))) \\
\alpha_1, \alpha_1, \alpha_3 & \mbox{ with corresponding context: } & x : (x' : (\Box : inf(s(x'')))) \\
\alpha_1, \alpha_1, \alpha_1 & \mbox{ with corresponding context: } & x : (x' : x'' : \Box) 
\end{eqnarray*}
Note that CDPs in chain candidates are assumed to be variable
disjoint, so the reoccurring variables have been renamed in the example.
It is easy to see that the box position is forbidden in all above
contexts, hence 
this position is also forbidden when $\Box$ is substituted by any term
$erase(t)$ because $\Box$ does not occur in any forbidden pattern. Hence, none of the
CDP chain candidates is a proper FP-CDP chain and thus according to Theorem \ref{thm_scp_sound}
it is sound to delete $\alpha_1$.
\end{example}

\section{Automated Synthesis of Forbidden Patterns}
\label{synth} 
 
In this section we are going to utilize the machinery of Sections
\ref{cdp} and \ref{scp},  
and in particular the simple context processor $SCP_n$,
in order to synthesize suitable forbidden patterns for
a given rewrite system $\mathcal{R}$. The basic idea is to
construct the CDPs of $\mathcal{R}$ assuming an empty set of
forbidden patterns $\Pi$ and then by an analysis with the $SCP_n$
processor synthesize the forbidden patterns needed to ensure $\Pi$-termination
of $\mathcal{R}$ on the fly.

More precisely, we analyze nested contexts obtained by sequences
of CDPs of bounded length (as in Definition \ref{def_scp}).
Let $c_1[\ldots[c_n[erase(t_n)]_{p_n}\ldots]_{p_1}$ be a term obtained by
this nested context analysis. In order to successfully apply 
the $SCP_n$ processor, position $p_1.\cdots.p_n$ must be forbidden
in this term. Hence, we synthesize a forbidden pattern 
$\langle c_1[\ldots[c_n[erase(t_n)]_{p_n}\ldots]_{p_1}, p_1.\cdots.p_n, h\rangle$,
that forbids exactly this position. By doing this for every sequence of
CDPs of length $n$ starting with the CDP corresponding to the context $c_1$, 
this CDP can be soundly deleted according to Theorem \ref{thm_scp_sound} provided
that the generated forbidden patterns are in $\Pi_{orth}$. However, 
forbidden patterns obtained this way might not be orthogonal to
the rewrite system and thus not be in $\Pi_{orth}$. In order to
overcome this problem, terms in the first component of synthesized 
forbidden patterns can be ``generalized'', i.e.\ linearized and
subterms at positions where overlaps with the rule system occur can
be replaced by fresh variables. By doing this the rewrite relation becomes
more restrictive (since the patterns match object terms more
easily). Moreover, since the patterns after this generalization
are orthogonal to $\mathcal{R}$, the $SCP_n$ processor is applicable
on the fly for simplifying 
the termination problems.

We provide an algorithmic schema for the forbidden pattern synthesis:
\begin{enumerate}
\item Compute $CDP(\mathcal{R})$ assuming an empty $\Pi$.
\item \label{start} Choose some CDP $s_1 \rightarrow t_1 \; [c_1]$.
\item For all CDP sequences 
\begin{equation*}
s_1 \rightarrow t_1 [c_1[\Box]_{p_1}], s_2 \rightarrow t_2 [c_2[\Box]_{p_2}], 
\ldots, s_{n} \rightarrow t_n [c_n[\Box]_{p_n}]
\end{equation*}
\begin{enumerate}
\item If position $p_1.\cdots.p_n$ is allowed in $c_1[\ldots[c_n[erase(t_n)]_{p_n}\ldots]_{p_1}$, 
\begin{enumerate}
\item create a forbidden pattern $\langle c_1[\ldots[c_n[erase(t_n)]_{p_n}\ldots]_{p_1}, p_1.\cdots.p_n, h\rangle$.
\item Generalize $\langle
  c_1[\ldots[c_n[erase(t_n)]_{p_n}\ldots]_{p_1}, p_1.\cdots.p_n,
  h\rangle$ so that it is orthogonal to $\mathcal{R}$, obtaining
  $\langle u, o, \lambda \rangle$. 
\item Add $\langle u, o, \lambda \rangle$ to $\Pi$.
\end{enumerate}
\end{enumerate}
\item \label{delete} Delete the CDP $s_1 \rightarrow t_1 \; [c_1]$ and continue the $\Pi$-termination analysis (e.g.\ at Stage \ref{start}).
\end{enumerate}
 
\begin{example}
Consider a CDP problem $(\mathcal{P}, \mathcal{R}, \Pi, T)$ where
$\mathcal{P} = \{a^{\#} \rightarrow a^{\#} [f(\Box)]\}$, 
$\mathcal{R} = \{a \rightarrow f(a)\}$ and 
$\Pi = \emptyset$
(cf.\ also Example \ref{ex_analyzing_contexts}).
An $SCP_2$ processor
encounters e.g.\ the term $f(f(erase(a^{\#})) = f(f(a))$. Thus, a
forbidden pattern $\pi = \langle f(f(a)), 1.1, h \rangle$
could be used. This forbidden pattern is orthogonal
to $\mathcal{R}$, hence there is no need to generalize it.
Indeed, when this forbidden pattern is used, there is
no infinite FP-CDP chain.
\end{example}

Usually one wants to restrict the shape of the generated
patterns for instance by demanding that all forbidden
patterns contain allowed redexes and do not overlap (each other), 
in order to ensure that 
$\Pi$-normal forms are normal forms (w.r.t.\ $\mathcal{R}$);
then termination of $\rightarrow_{\Pi}$ implies weak termination 
of $\rightarrow_{\mathcal{R}}$.

A second choice for restrictions on the shape of forbidden
patterns might be 
\emph{canonical} forbidden patterns as defined 
in \cite{wrs09}[Definition 4]. 

Synthesis of forbidden patterns adhering to these syntactical
restrictions can be 
done
analogously to the way
patterns orthogonal to $\mathcal{R}$ are synthesized. Namely, 
by generalizing the forbidden patterns to make them compatible
with syntactical constraints immediately after their 
creation.

\begin{example}
\label{ex_fp_creat}
Consider the CDP problem of $\mathcal{R}$ of Example \ref{ex_2nd}and the contextual
dependency pair 
\begin{equation*}
inf^{\#}(x) \rightarrow inf^{\#}(s(x)) [x : \Box]
\end{equation*}
Applying an $SCP_2$ processor we get a term $x : x' : inf(s(x'))$, which 
needs to be generalized since it is not linear
and thus not canonical (because not simple) and not in $\Pi_{orth}$ 
(hence a termination proof with
the context processor would not be possible). Instead we
linearize the term obtaining $x : y : inf(s(z))$ which we
can use as canonical forbidden pattern. Indeed,  
$\mathcal{R}$ is $\Pi$-terminating when choosing $\Pi = \langle x : y
: inf(s(z)), 2.2, h \rangle$.  
\end{example}

As an alternative to the on-the-fly generation of forbidden
patterns during the termination analysis with $SCP_n$ processors, 
in some cases an (iterated) two phase process might be more efficient.
There, Stage \ref{delete} of the above algorithm scheme is
not carried out, i.e.\ no CDPs are deleted after the generation
of forbidden patterns. Instead, in phase 2, the termination analysis starts
from scratch using the generated forbidden patterns. If it fails
new forbidden patterns are generated and termination is analyzed again
afterwards. This sequence of (separated) generation of forbidden patterns
and termination analysis continues until termination is proved.

While at first glance the two phase approach seems to be less
efficient than the on-the-fly generation of forbidden patterns
during the termination analysis, it has an important advantage.
In the phase of the generation of forbidden patterns an arbitrary
subset of CDPs can be used for the synthesis of forbidden patterns.
Since termination is proved separately, this does not affect the
soundness of the approach. 
The concrete advantages of this approach are the following.
\begin{itemize}
\item For the termination analysis one is not restricted to the CDP
framework. One can for instance use the transformation
of \cite{wrs09}.
\item When disregarding structural dependency pairs during the
synthesis of forbidden patterns, the generated patterns are more
intuitive, simpler and often suffice to obtain termination.
\item When using the CDP framework, the generated (stable) forbidden patterns
can be used to compute the concrete set of CDPs.
\item The generation of forbidden patterns is more fine-grained, since
not all sequences of CDPs of a given length are considered in the $SCP_n$
processor, but only those contained in the specified subset of CDPs (which could for
instance be specified by a human in a semi-automatic synthesis process). This results
in fewer created forbidden patterns that might still be sufficient to
yield $\Pi$-termination. 
\end{itemize}

\begin{example}
In Example \ref{ex_fp_creat} exactly the only non-structural dependency
pair is used. Using the two phase approach the according forbidden pattern
is found fully automatically. 
\end{example}

\section{Implementation and Evaluation}

We implemented the CDP framework and the context processors in the
termination tool VMTL (cf.\ \cite{rta_vmtl}). 
In order to evaluate the practical power of this approach we tested this
implementation on the TRSs of the outermost category of the TPDB%
\footnote{
The termination problem database, available at \texttt{http://termcomp.uibk.ac.at/}
}.
Since outermost rewriting is a special case of rewriting with forbidden
patterns (in particular rewriting with forbidden $b$-patterns), the
CDP approach is applicable to these systems. In our test run 291 TRS were
evaluated, 158 of which were proven to be outermost non-terminating in the
termination competition 2008 (\cite{termcomp}). Table \ref{tab_outermost_benchmarks}
shows the results of VMTL on the test set. At the time of writing, VMTL does not
support non-termination analysis of outermost TRSs. Hence, Table \ref{tab_outermost_benchmarks}
indicates only the positive results of VMTL and various other termination tools
tested on the same set of examples. We cite the results of the termination competition
2008 since the then most powerful tool (regarding positive termination proofs) Jambox
did not participate in the subsequent years.
The participating tools were
\begin{itemize}
\item AProVE (\cite{Aprove}), which proves outermost termination by transforming TRSs 
such that (innermost) termination implies outermost termination of the original TRS. 
The transformations used are the ones from Raffelsieper et.\ al.\ (\cite{raffelsieper})
and Thiemann (\cite{thiemann}).
\item TrafO (\cite{raffelsieper}), which proves outermost termination
by 
using the transformation of Raffelsieper et.\ al.\ (\cite{raffelsieper})
and analyzing the resulting TRSs with Jambox (\cite{jambox}).
\item ``Jambox goes out'', which transforms
TRSs into a context-sen\-si\-tive ones, such that termination of the
latter implies
outermost termination of the former (cf.\ \cite{rta_outermost}).    
\end{itemize}
TTT2 (\cite{ttt2}) participated in the outermost category of the termination competition
2008 but was specialized (exclusively) on disproving outermost termination. Hence, 
it is not included in Table \ref{tab_outermost_benchmarks}.

\begin{table}
\begin{equation*}
\begin{tabular}{|c|c|c|c|}
\hline
\textbf{VMTL Simple} & \textbf{AProVE} & \textbf{TrafO} & \textbf{Jambox goes out} \\
\hline
$33$ & $27$ & $46$ & $72$ \\ \hline
\end{tabular}
\end{equation*}
\caption{Number of successful outermost termination proofs of various systems.}
\label{tab_outermost_benchmarks}
\end{table}

VMTL used the simple context processor $SCP_n$ with $n = 3$
for the analysis of contexts and reduction pair processors based
on polynomial interpretations as well as a dependency graph processor. 

We would like to stress at this point that the performance of VMTL
vastly improves when using more clever ways of context-analysis like
the one described in \cite{wst10}. Using the \emph{context processor
based on tree automata} introduced there outermost termination of $60$
examples can be automatically verified by VMTL. However, even then
the transformation approach of ``Jambox goes out'' has the edge over
VMTL in proving outermost termination.
We believe that the reasons for this are twofold.

First, the use of structural dependency pairs adds significant complexity
to the initial CDP problems. In particular for some outermost terminating TRSs
where VMTL failed to find an outermost termination proof, we observed
that the simplified CDP problems obtained at the end of failed proof attempts
consisted of structural CDPs only (in their first
component). 
The second reason for the lack of power of VMTL compared to Jambox is the
use of $\Pi_{orth}$ in the context processors. By excluding certain
forbidden patterns in the context analysis performed by our context processors
the power is reduced.

Addressing both of these problems seems like an interesting and promising
way to improve the CDP framework and make it even more competitive in the
future.

\section{Conclusion and Related Work}

We introduced a modified version of the dependency pair framework where
dependency pairs are enriched by an additional component that is best
understood as the calling context of the recursive function
the dependency pair originated from. This contextual dependency pair
framework enables us to reason about termination of rewriting
incorporating many forms of context-dependency resp.\ context-sensitivity.

In that sense the context-sensitive dependency pair framework of \cite{lpar08}
might be seen as a specialized and optimized version of the CDP framework
where the contextual information is incorporated into the dependency
pairs directly without explicitly having these contexts attached to the
dependency pairs. However, in the case of context-sensitive rewriting 
this was possible mainly because of the simplicity and stability
of context-sensitive restrictions (note that if context-sensitivity
is expressed by forbidden pattern restrictions as in \cite{wrs09}, all
resulting forbidden patterns are stable and $\Pi = \Pi_{orth}$). 
In the presence of more sophisticated
context restrictions the CDP framework 
appears to
be advantageous and more general 
because of the explicit reference to the strategic restrictions in the notion
of chains.

In the case of context-sensitive rewriting there are also other even more
general formulations of dependency pairs and context-sensitive dependency
pair frameworks (cf.\ \cite{rauldiss,wrla10guti}). There, dependency pairs are
allowed to be 
\emph{collapsing}. 
Thus, the use of structural dependency
pairs can be avoided. Since structural dependency pairs are a major
source of complication and practical limitation in our contextual dependency
pair framework, it might be a promising direction of future research
to use collapsing contextual dependency pairs as well.  

In order to prove termination within the CDP framework we introduced
the simple context processor $SCP_n$. This processor analyzes sequences
of CDPs of bounded length for being proper FP-CDP chains and erases
a CDP if all chain candidates starting with this CDP cannot be proper
FP-CDP chains. 
Together with the CDP framework this processor yields an effective way
of proving termination of rewriting restricted with forbidden patterns.
Moreover, based on this processor we introduced a method to synthesize
forbidden patterns suitable for a given rewrite system on-the-fly
during the termination analysis. 

Regarding future work, we see several attractive directions. First, 
the power of termination analysis could be significantly
increased by using more clever methods of analyzing the (nested)
contexts of CDP chain candidates or more efficient ways to
represent CDP chain candidates. Some work has already been done in
this direction as reported in \cite{wst10}. There the nested contexts
of all possible sequences of CDPs are expressed finitely through the language
accepted by a certain tree automaton. Then it is checked whether in
every context of this language the hole position is forbidden by the
forbidden pattern restrictions, and if yes the CDP problem is simplified
accordingly.

Another direction of future research is finding larger subsets
of forbidden patterns for which Lemma \ref{lem_contexts} holds thus
enabling the use of larger subsets of forbidden patterns
in the context processors. 

Regarding the automated synthesis of forbidden patterns, building
upon the approach of Section \ref{synth} one of the challenges is
to generate small and intuitive sets of forbidden patterns. The two
phase approach described in Section \ref{synth} is already a first step
in this direction. Apart from that it might be interesting to use
more sophisticated methods of context analysis, such as the one based
on tree automata, for the generation of suitable forbidden patterns. 

\bibliographystyle{eptcs}

\end{document}